\documentclass[aip,jmp,amsmath,amssymb,reprint,onecolumn,longbibliography]{revtex4-2}

\usepackage{amsmath,amssymb,amsfonts}
\usepackage{amsthm}
\usepackage{mathrsfs}
\usepackage{graphicx}
\usepackage{booktabs}
\usepackage{multirow}
\usepackage{xcolor}
\numberwithin{equation}{section}

\DeclareMathOperator{\Tr}{Tr}
\DeclareMathOperator{\supp}{supp}
\DeclareMathOperator{\dist}{dist}

\newtheorem{theorem}{Theorem}[section]
\newtheorem{lemma}[theorem]{Lemma}
\newtheorem{proposition}[theorem]{Proposition}
\newtheorem{corollary}[theorem]{Corollary}
\theoremstyle{definition}
\newtheorem{definition}[theorem]{Definition}
\newtheorem{example}[theorem]{Example}

\begin{document}

\title{Entropy Stability under Block-Separable Constraints: Fixed-Weight Rank-Lattice Envelopes and Convexification}

\author{Hassan Nasreddine}
\email{hassan.nasreddine@fortisec.net}
\affiliation{FortiSec Systems Inc.}

\begin{abstract}
For a block-diagonal state with prescribed block weights and dimensions, we
determine the minimum von Neumann entropy excess at each trace-norm distance
from the blockwise-pure states.  The one-block extremal is the classical
reverse-Fano profile; the new problem is its constrained assembly when the
weights $q_i$ and dimensions $n_i$ cannot be redistributed.  The resulting
value function $\Phi_{q,\mathbf n}(D)$ is attained on the full admissible
distance interval and is the lower envelope of finitely many rank-lattice
edges.  Every minimizing distance allocation has at most one block between
consecutive integer ranks, and the active edge can switch nonmonotonically
with $D$.  We also determine the lower convex envelope of $\Phi_{q,\mathbf n}$ explicitly.
It is the polygonal curve through synchronized rank vertices and admits a
one-parameter dual formula.  The exact nonconvex envelope differs from this
convex surrogate by at most $q_\diamond\log(5/4)$, with a sharper
stage-dependent bound; the constant is optimal.  For binomial charge-sector
weights the uniform error is $O(L^{-1/2})$.  An exact entropy-gap decomposition
places these results in a broader block-separable stability framework and
yields a symmetry-sector purity certificate.  In a three-qubit phase-damping
example the fixed-weight certificate remains nontrivial where the universal
classical entropy--error envelope is already vacuous.
\end{abstract}


\maketitle

\section{Introduction}
\label{sec:intro}

Let $\mathcal H$ be a finite-dimensional Hilbert space and
$\mathcal S(\mathcal H)$ its state space. We consider entropy minimization on
convex sets compatible with an orthogonal decomposition
\[
 \mathcal H=\bigoplus_{i=1}^r \mathcal H_i.
\]
The corresponding states have the form
\[
 \rho=\bigoplus_{i=1}^r p_i\rho_i,
 \qquad p\in\Delta_r,
 \qquad \rho_i\in\mathcal S(\mathcal H_i),
\]
and their entropy splits as
\begin{equation}
\label{eq:intro-split}
 S(\rho)=H(p)+\sum_{i=1}^r p_iS(\rho_i).
\end{equation}
The identity is elementary, but it separates two geometrically different
sources of entropy: uncertainty in the block populations and mixedness inside
the blocks.

The main result concerns fixed block populations.  With unrestricted internal
states, the minimizers are blockwise pure.  We determine the least entropy
excess at a prescribed trace-norm distance $D$ from this set over the entire
admissible range.  The one-block problem is classical: it is the sharp
minimum-entropy problem at fixed prediction error, equivalently fixed distance
from the pure-state set.  The multi-block problem retains that one-block
profile while prescribing the weights $q_i$ and dimensions $n_i$.  Its exact
value function is a finite lower envelope of rank-lattice edges; every
minimizing distance allocation has at most one block between consecutive
integer ranks, and the active edge can change nonmonotonically with $D$.

The classical one-block profile goes back to the reverse-Fano literature,
including Kovalevsky~\cite{Kovalevsky1968}, Tebbe--Dwyer~\cite{TebbeDwyer1968},
Feder--Merhav~\cite{FederMerhav1994}, and the modern majorization treatment of
Rioul~\cite{Rioul2023}.  Ho--Yeung~\cite{HoYeung2010}, Ho--Verd\'u~\cite{HoVerdu2010},
Prelov~\cite{Prelov2014}, Sason--Verd\'u~\cite{SasonVerdu2018}, and
Hu--Xing~\cite{HuXing2016} treat related entropy--error and entropy--distance extremal problems.
Additional posterior information is used by Prasad~\cite{Prasad2015} and,
under a prescribed MAP confusion matrix, by Hled\'ik, Sokolowski and
Tka\v{c}ik~\cite{HledikSokolowskiTkacik2019}.  In such conditional problems an
observation or decode class may be split into several outputs, which
convexifies the one-posterior profile.  A fixed block here contains one density
operator and its weight cannot be split or reallocated.  The stagewise
concave profile is therefore retained, together with its upward rank kinks;
this is the source of the nonconvex allocation geometry in
Section~\ref{sec:weight-profile}.

The lower convex envelope of the exact value function nevertheless has a
simple closed form.  Section~\ref{sec:convexification} shows that it is the
polygonal interpolation through synchronized rank vertices, gives an exact
one-parameter dual representation, and bounds the nonconvexity gap by the
sharp constant $q_\diamond\log(5/4)$.

For varying populations, the direct-sum identity also separates marginal and
within-block stability.  On a fixed-support polytope, a nonisolated entropy
minimum has a positive first-order slope.  Opening a zero eigenvalue produces
a $t\log(1/t)$ term, while curved active constraints can create a quadratic
minimum despite the negative entropy Hessian.  Degenerate constraints can
produce higher-order laws.  Sections~\ref{sec:fixed-support}--\ref{sec:composition}
record these local mechanisms and their composition.

A complementary analysis without block diagonality is developed in the
companion manuscript~\cite{NasreddineRankStratified2026}, using rank
stratification and Schur-defect coordinates to separate coherent support
motion from genuine rank creation.  The present paper assumes the block
decomposition and uses it to obtain an exact fixed-weight value function over
the full distance range.

Block diagonal structures arise under superselection rules, symmetry
reduction, charge-sector decompositions, and projective dephasing; see, for
example, Bartlett, Rudolph and Spekkens~\cite{BartlettRudolphSpekkens2007}.
The decomposition~\eqref{eq:intro-split} then separates fluctuations between
sectors from mixedness within sectors. When the charge statistics are known,
the fixed-weight envelope gives a state-level certificate for sectorwise
purity. Section~\ref{sec:symmetry} illustrates this with a three-qubit
phase-damping example whose within-sector entropy is nonzero.

The paper is organized as follows. Section~\ref{sec:block-reduction} gives the
gap decomposition for general block-separable constraints. Sections
\ref{sec:pure-global} and \ref{sec:weight-profile} develop the one-block
profile and the exact fixed-weight envelope, including rank-lattice geometry
and the two-qubit switching transition. Section~\ref{sec:convexification}
derives the convex surrogate and its sharp error bound. Sections
\ref{sec:fixed-support}--\ref{sec:composition} treat local stability
mechanisms and their composition, and Section~\ref{sec:symmetry} applies the
fixed-weight results to symmetry sectors. Section~\ref{sec:comparison}
compares the bounds with continuity, Pinsker, and classical entropy--error
inequalities.

The analysis is finite dimensional. Standard convex-analysis and
second-order optimization facts are used in the form found, for example, in
Rockafellar~\cite{Rockafellar1970} and Bonnans--Shapiro~\cite{BonnansShapiro2000}.

\section{Preliminaries}
\label{sec:prelim}

All logarithms are natural. For a state $\rho$,
\[
 S(\rho):=-\Tr(\rho\log\rho),
\]
with the convention $0\log0=0$. For $p\in\Delta_r$,
\[
 H(p):=-\sum_{i=1}^r p_i\log p_i.
\]
We use the trace norm $\|X\|_1=\Tr|X|$ and write
\[
 \dist_1(x,M):=\inf_{y\in M}\|x-y\|_1.
\]
Thus $\dist_1$ denotes the unhalved trace-norm distance.  When the
conventional trace distance between two states is needed, we write
$T(\rho,\sigma):=\tfrac12\|\rho-\sigma\|_1$.

\begin{lemma}[Entropy of a direct sum]
\label{lem:direct-sum}
If $\rho=\bigoplus_i p_i\rho_i$, with $p\in\Delta_r$ and normalized
$\rho_i$, then
\[
 S(\rho)=H(p)+\sum_i p_iS(\rho_i).
\]
\end{lemma}

\begin{proof}
The nonzero eigenvalues of the $i$th block are $p_i\lambda_{ij}$, where
$\lambda_{ij}$ are the eigenvalues of $\rho_i$. Substitution in the spectral
formula for entropy gives the identity.
\end{proof}

\begin{lemma}[Strict concavity]
\label{lem:strict-concavity}
The Shannon entropy is strictly concave on $\Delta_r$, and the von Neumann
entropy is strictly concave on the state space: for distinct states $\rho$ and
$\sigma$ and $0<t<1$,
\[
 S((1-t)\rho+t\sigma)>(1-t)S(\rho)+tS(\sigma).
\]
\end{lemma}

Strict concavity does not require a unique minimizer: distinct minimizing
extreme points may exist, but a nontrivial convex combination of them cannot
also minimize the entropy.

\section{Block-separable constraints and the gap decomposition}
\label{sec:block-reduction}

\begin{definition}[Block-separable convex constraint]
Let $\Pi\subset\Delta_r$ be compact and convex, and let
$\mathcal C_i\subset\mathcal S(\mathcal H_i)$ be nonempty compact convex sets.
We define
\[
 \mathcal C:=
 \left\{\bigoplus_{i=1}^r p_i\rho_i:
 p\in\Pi,\ \rho_i\in\mathcal C_i\right\}.
\]
For each block set
\begin{equation}
\label{eq:block-minima}
 m_i:=\min_{\tau\in\mathcal C_i}S(\tau),
 \qquad
 \mathcal M_i:=\arg\min_{\tau\in\mathcal C_i}S(\tau),
\end{equation}
and define
\begin{equation}
\label{eq:effective-F}
 F(p):=H(p)+\sum_{i=1}^r p_im_i,
 \qquad p\in\Pi.
\end{equation}
Let
\[
 F_*:=\min_{p\in\Pi}F(p),
 \qquad
 P_*:=\arg\min_{p\in\Pi}F(p).
\]
\end{definition}

\begin{theorem}[Gap decomposition and minimizer structure]
\label{thm:exact-reduction}
For the constraint set above,
\[
 S_{\min}:=\min_{\rho\in\mathcal C}S(\rho)=F_*.
\]
Moreover, a state $\sigma=\bigoplus_i q_i\sigma_i$ is an entropy minimizer if
and only if $q\in P_*$ and $\sigma_i\in\mathcal M_i$ for every $i$ with
$q_i>0$. Every admissible $\rho=\bigoplus_i p_i\rho_i$ satisfies
\begin{equation}
\label{eq:exact-gap}
 S(\rho)-S_{\min}
 =F(p)-F_*+\sum_{i=1}^r p_i\bigl(S(\rho_i)-m_i\bigr).
\end{equation}
The marginal and conditional contributions on the right are separately
nonnegative.
\end{theorem}

\begin{proof}
By Lemma~\ref{lem:direct-sum},
\[
 S(\rho)=H(p)+\sum_i p_iS(\rho_i)
 \ge H(p)+\sum_i p_im_i=F(p)\ge F_*.
\]
For fixed $p$, equality in the first inequality holds exactly when every
active block state minimizes entropy on $\mathcal C_i$. Minimizing over $p$
gives the characterization, and subtracting $F_*$ gives
\eqref{eq:exact-gap}.
\end{proof}

\begin{corollary}[Marginal minimizers of a polytope are vertices]
\label{cor:vertices}
If $\Pi$ is a polytope, then every element of $P_*$ is a vertex of $\Pi$.
Hence $P_*$ is finite.
\end{corollary}

\begin{proof}
The functional $F$ is strictly concave because $H$ is strictly concave and
the second term in~\eqref{eq:effective-F} is affine. If a minimizer $q$ were
a nontrivial convex combination $q=(1-t)p+tr$ of two distinct points of
$\Pi$, strict concavity would give
\[
 F(q)>(1-t)F(p)+tF(r)\ge F_*,
\]
a contradiction.
\end{proof}

\begin{lemma}[A blockwise distance estimate]
\label{lem:block-distance}
Let $\rho=\bigoplus_i p_i\rho_i$ and $\sigma=\bigoplus_i q_i\sigma_i$.
Then
\begin{equation}
\label{eq:block-distance}
 \|\rho-\sigma\|_1
 \le \|p-q\|_1+\sum_iq_i\|\rho_i-\sigma_i\|_1.
\end{equation}
\end{lemma}

\begin{proof}
On each orthogonal block,
\[
 p_i\rho_i-q_i\sigma_i=(p_i-q_i)\rho_i+q_i(\rho_i-\sigma_i).
\]
The trace norm is additive across orthogonal blocks and $\|\rho_i\|_1=1$.
\end{proof}

\section{One-block entropy profile}
\label{sec:pure-global}

Let $\mathrm{Pure}(\mathcal H)$ denote the rank-one states on $\mathcal H$.
We first solve the one-block problem over its full distance range.

\begin{lemma}[Distance to the pure-state set]
\label{lem:pure-distance}
Let $\rho$ be a state and let $\lambda_{\max}(\rho)$ be its largest eigenvalue.
Then
\begin{equation}
\label{eq:pure-distance}
 \dist_1(\rho,\mathrm{Pure})=2\bigl(1-\lambda_{\max}(\rho)\bigr).
\end{equation}
\end{lemma}

\begin{proof}
Let $P=|\psi\rangle\langle\psi|$. Applying the two-outcome measurement
$\{P,I-P\}$ and using contractivity of the trace norm under the measurement channel gives
\[
 \|\rho-P\|_1\ge2\bigl(1-\Tr(\rho P)\bigr)
 \ge2(1-\lambda_{\max}(\rho)).
\]
If $P$ projects onto an eigenvector for $\lambda_{\max}$, then $\rho-P$ is
diagonal in the same eigenbasis and its trace norm is exactly the right-hand
side.
\end{proof}

For $k\ge1$ define
\begin{equation}
\label{eq:delta-k}
 \delta_k:=2\left(1-\frac1k\right),
 \qquad \delta_1=0.
\end{equation}
For a block of dimension $n$, define $\beta_n:[0,\delta_n]\to[0,\log n]$ as
follows. If $\delta_k\le d\le\delta_{k+1}$ for some
$1\le k\le n-1$, set
\begin{equation}
\label{eq:beta-profile}
 p:=1-\frac d2,
 \qquad r:=1-kp,
 \qquad
 \beta_n(d):=-kp\log p-r\log r,
\end{equation}
with $0\log0=0$. At a common endpoint either adjacent formula gives the same
value. For $n=1$ we set $\beta_1(0)=0$.

\begin{theorem}[One-block entropy profile]
\label{thm:one-block-profile}
Let $\rho$ act on an $n$-dimensional Hilbert space and put
$d=\dist_1(\rho,\mathrm{Pure})$. Then
\begin{equation}
\label{eq:one-block-profile}
 S(\rho)\ge\beta_n(d),
 \qquad 0\le d\le\delta_n,
\end{equation}
and equality is attained for every $d$. More precisely, on the stage
$\delta_k\le d\le\delta_{k+1}$ an entropy-minimizing spectrum is
\begin{equation}
\label{eq:saturated-spectrum}
 (\underbrace{p,\ldots,p}_{k\ \mathrm{times}},r,0,\ldots,0),
 \qquad p=1-d/2,\quad r=1-kp.
\end{equation}
Thus the endpoint $d=\delta_k$ is the uniform distribution on a
$k$-dimensional subspace. On the interior of each stage,
\begin{equation}
\label{eq:beta-derivatives}
 \beta_n'(d)=\frac{k}{2}\log\frac{p}{r}>0,
 \qquad
 \beta_n''(d)=-\frac{k}{4}\left(\frac1p+\frac{k}{r}\right)<0.
\end{equation}
In particular, each stage is strictly increasing and strictly concave.
\end{theorem}

\begin{proof}
By Lemma~\ref{lem:pure-distance}, the largest eigenvalue is
$p=1-d/2$. Choose $k$ so that $1/(k+1)\le p\le1/k$ and put
$r=1-kp$. Every eigenvalue of $\rho$ is at most $p$. The vector in
\eqref{eq:saturated-spectrum} therefore majorizes the eigenvalue vector of
$\rho$: for the first $\ell\le k$ partial sums it gives $\ell p$, and after
the $(k+1)$st component its partial sum is already one. Schur concavity of
Shannon entropy~\cite{MarshallOlkinArnold2011} yields
$S(\rho)\ge\beta_n(d)$. The spectrum~\eqref{eq:saturated-spectrum} realizes
equality and has the prescribed largest eigenvalue. Differentiating
\eqref{eq:beta-profile} on the interior of a stage gives
\eqref{eq:beta-derivatives}.
\end{proof}

The spectrum in~\eqref{eq:saturated-spectrum} is the classical
minimum-entropy distribution at fixed prediction error $e=d/2$. The associated
sharp Shannon reverse-Fano bound goes back to Kovalevsky~\cite{Kovalevsky1968};
see also Tebbe--Dwyer~\cite{TebbeDwyer1968},
Feder--Merhav~\cite{FederMerhav1994}, and the majorization treatment of
Rioul~\cite{Rioul2023}. Allowing posterior distributions to be mixed with
unrestricted observation weights gives the classical convexified conditional
lower envelope. Section~\ref{sec:comparison} compares that setting with the
fixed block weights used below.

The profile $\beta_n$ is continuous but not globally concave.  At every
interior rank point $\delta_k$, $2\le k\le n-1$, the adjacent stage
derivatives satisfy
\begin{equation}
\label{eq:beta-kink}
 \beta_n'(\delta_k^-)=0,
 \qquad
 \beta_n'(\delta_k^+)=+\infty.
\end{equation}
Thus the rank lattice produces upward kinks between the strictly concave
stages.  These kinks disappear under the classical convexification but remain
in the fixed-weight problem below.

The first stage of $\beta_n$ is the familiar binary expression:
\begin{equation}
\label{eq:binary-first-stage}
 \beta_n(d)=h_2(d/2),\qquad 0\le d\le1,
\end{equation}
for every $n\ge2$.

\section{Fixed block populations and the weight-profile envelope}
\label{sec:weight-profile}

Fix $q=(q_1,\ldots,q_r)$ with $q_i>0$ and write
$n_i=\dim\mathcal H_i$. Set
\[
 \mathcal C_q:=
 \left\{\bigoplus_{i=1}^r q_i\rho_i:
 \rho_i\in\mathcal S(\mathcal H_i)\right\},
\]
with minimizer set
\[
 \mathcal M_q:=
 \left\{\bigoplus_{i=1}^r q_iP_i:
 P_i\ \text{rank one on }\mathcal H_i\right\}.
\]
Blocks with $n_i=1$ are automatically pure. Define
\begin{equation}
\label{eq:Dmax}
 D_{\max}:=\sum_{i=1}^rq_i\delta_{n_i}
 =2\sum_{i=1}^rq_i\left(1-\frac1{n_i}\right).
\end{equation}
For $0\le D\le D_{\max}$ define
\begin{equation}
\label{eq:Phi-variational}
 \Phi_{q,\mathbf n}(D)
 :=\min\left\{
 \sum_{i=1}^r q_i\beta_{n_i}(d_i):
 0\le d_i\le\delta_{n_i},\quad
 \sum_{i=1}^r q_id_i=D
 \right\}.
\end{equation}
This is a finite-dimensional resource-allocation problem whose structure is
described below.

\begin{theorem}[Fixed-population envelope]
\label{thm:global-envelope}
For every $\rho=\bigoplus_iq_i\rho_i\in\mathcal C_q$, let
\[
 D:=\dist_1(\rho,\mathcal M_q).
\]
Then
\begin{equation}
\label{eq:global-envelope}
 S(\rho)-S_{\min}\ge\Phi_{q,\mathbf n}(D),
 \qquad
 S_{\min}=H(q),
 \qquad 0\le D\le D_{\max}.
\end{equation}
For every $D$ in this interval there exists a state in $\mathcal C_q$ that
attains equality. Thus $\Phi_{q,\mathbf n}$ is the entropy minimum at
prescribed trace-norm distance from $\mathcal M_q$.
\end{theorem}

\begin{proof}
For $d_i:=\dist_1(\rho_i,\mathrm{Pure}(\mathcal H_i))$, orthogonality of the
blocks gives
\begin{equation}
\label{eq:D-sum}
 D=\sum_iq_id_i.
\end{equation}
Since the populations are fixed,
\[
 S(\rho)-S_{\min}=\sum_iq_iS(\rho_i).
\]
Theorem~\ref{thm:one-block-profile} therefore gives
\[
 S(\rho)-S_{\min}
 \ge\sum_iq_i\beta_{n_i}(d_i)
 \ge\Phi_{q,\mathbf n}(D).
\]
The feasible set in~\eqref{eq:Phi-variational} is compact, so a minimizing
vector $(d_i)$ exists. For each $i$, choose a block state with the saturated
spectrum~\eqref{eq:saturated-spectrum} corresponding to $d_i$. Each
one-block inequality is then an equality, and~\eqref{eq:D-sum} gives total
distance $D$. This constructs an equality state for every admissible $D$.
\end{proof}

\begin{proposition}[Comparison with the universal entropy--error convex envelope]
\label{prop:feder-merhav-comparison}
Let $M:=\max_i n_i$ and let $\varphi_M^*$ denote the piecewise-linear convex
function joining the points
\[
 \left(1-\frac1k,\log k\right),\qquad k=1,\ldots,M.
\]
Then
\begin{equation}
\label{eq:FM-comparison}
 \Phi_{q,\mathbf n}(D)\ge\varphi_M^*(D/2),
 \qquad 0\le D\le D_{\max}.
\end{equation}
The right-hand side is the classical conditional entropy--error convex
envelope associated with the Kovalevsky reverse-Fano bound
\cite{Kovalevsky1968,Rioul2023}. Equality may occur for special weight
profiles, but prescribed $q_i$ generally prevent arbitrary convexification.
\end{proposition}

\begin{proof}
For each block, the one-distribution curve
$e\mapsto\beta_{n_i}(2e)$ lies above the convex minorant
$\varphi_M^*(e)$ on its admissible interval. Therefore, for every feasible
$(d_i)$ in~\eqref{eq:Phi-variational},
\[
 \sum_iq_i\beta_{n_i}(d_i)
 \ge\sum_iq_i\varphi_M^*(d_i/2)
 \ge\varphi_M^*\!\left(\frac12\sum_iq_id_i\right)
 =\varphi_M^*(D/2),
\]
where the second inequality is Jensen's inequality. Minimizing over
$(d_i)$ proves~\eqref{eq:FM-comparison}.
\end{proof}

The variational definition~\eqref{eq:Phi-variational} can be reduced to a
finite explicit lower envelope. For a rank vector
$\mathbf k=(k_1,\ldots,k_r)$ with $1\le k_i\le n_i$, define
\begin{equation}
\label{eq:rank-vertex}
 D_{\mathbf k}:=\sum_iq_i\delta_{k_i},
 \qquad
 E_{\mathbf k}:=\sum_iq_i\log k_i.
\end{equation}
These are respectively the distance and conditional entropy when block $i$ is
uniform on a $k_i$-dimensional subspace. If $k_j<n_j$, define the edge
interval
\begin{equation}
\label{eq:edge-interval}
 I_{\mathbf k,j}
 :=\left[D_{\mathbf k},
 D_{\mathbf k}+q_j(\delta_{k_j+1}-\delta_{k_j})\right]
\end{equation}
and, for $D\in I_{\mathbf k,j}$,
\begin{equation}
\label{eq:edge-arc}
 \Gamma_{\mathbf k,j}(D)
 :=E_{\mathbf k}
 +q_j\left[
 \beta_{n_j}\!\left(
 \delta_{k_j}+\frac{D-D_{\mathbf k}}{q_j}
 \right)-\log k_j\right].
\end{equation}

If every $n_i=1$, then $D_{\max}=0$ and the problem is trivial. Otherwise:

\begin{theorem}[Rank-lattice edge representation]
\label{thm:rank-lattice}
For every $D\in[0,D_{\max}]$,
\begin{equation}
\label{eq:rank-envelope}
 \Phi_{q,\mathbf n}(D)
 =\min_{\substack{\mathbf k,\ j:\ k_j<n_j\\D\in I_{\mathbf k,j}}}
 \Gamma_{\mathbf k,j}(D).
\end{equation}
Moreover, every minimizing vector $(d_i)$ in
\eqref{eq:Phi-variational} has at most one coordinate outside the finite rank
lattice $\{\delta_1,\ldots,\delta_{n_i}\}$.  Consequently, every state that
attains equality in Theorem~\ref{thm:global-envelope} has, up to unitary
rotations within the blocks, all but at most one block uniform on an
integer-dimensional subspace.  The remaining block, if present, has spectrum
\[
 (\underbrace{p,\ldots,p}_{k_j\ \mathrm{times}},r,0,\ldots,0)
\]
for one stage between ranks $k_j$ and $k_j+1$.  Hence
$\Phi_{q,\mathbf n}$ is the pointwise lower envelope of the edges
\eqref{eq:edge-arc}.  The minimizing edge may change as $D$ varies, so the
rank allocation need not be nested in $D$.
\end{theorem}

\begin{proof}
Partition each interval $[0,\delta_{n_i}]$ into its one-block stages
$[\delta_k,\delta_{k+1}]$.  On a product of such stages, the objective in
\eqref{eq:Phi-variational} is concave by
Theorem~\ref{thm:one-block-profile}.  Its intersection with
$\sum_iq_id_i=D$ is a compact polytope, so a minimum can be attained at an
extreme point.  An extreme point of a box cut by one independent affine
hyperplane has at most one coordinate strictly between its two box bounds:
if $d_i$ and $d_j$ were both interior, then for small $t$ the perturbation
\[
 d_i\mapsto d_i+t/q_i,
 \qquad
 d_j\mapsto d_j-t/q_j
\]
and its negative would remain feasible.  Thus an extreme minimizer lies on a
rank-lattice edge.  Its value is precisely \eqref{eq:edge-arc}, and minimizing
over all rank vectors and active edges gives \eqref{eq:rank-envelope}.
Conversely, the spectra in Theorem~\ref{thm:one-block-profile} realize every
such edge, so the reduction is exact.

It remains to characterize arbitrary minimizers.  Suppose a minimizing vector
$(d_i)$ had two coordinates $d_i,d_j$ strictly inside their respective rank
stages.  Along the same constraint-preserving perturbation above, both
coordinates remain in those stages for sufficiently small $|t|$.  The
restriction of the objective to this line is then strictly concave, because
$\beta_{n_i}$ and $\beta_{n_j}$ are strictly concave in the interiors of their
stages.  An interior point of that line cannot minimize a strictly concave
function, a contradiction.  Hence every minimizing vector has at most one
non-lattice coordinate.

Finally, equality in Theorem~\ref{thm:global-envelope} forces equality in each
one-block bound and forces $(d_i)$ to minimize
\eqref{eq:Phi-variational}.  Strict Schur concavity in the proof of
Theorem~\ref{thm:one-block-profile} then gives the stated block spectra, up to
unitary rotations.
\end{proof}

The extreme-point step is the standard minimum principle for a concave
function on a polytope; see, for example, Rockafellar~\cite{Rockafellar1970}.
Here it is combined with the entropy profile in
Theorem~\ref{thm:one-block-profile} and the prescribed weights $q_i$ to
obtain the finite family of arcs in~\eqref{eq:rank-envelope}.

\begin{corollary}[Binary first branch]
\label{cor:binary-first-branch}
Assume at least one block has dimension at least two and define
\[
 q_*:=\min\{q_i:n_i\ge2\}.
\]
Then
\begin{equation}
\label{eq:binary-branch}
 \Phi_{q,\mathbf n}(D)
 =q_*h_2\!\left(\frac{D}{2q_*}\right),
 \qquad 0\le D\le q_*.
\end{equation}
Hence
\begin{equation}
\label{eq:small-D-asymptotic}
 \Phi_{q,\mathbf n}(D)
 =\frac D2\log\frac{2q_*}{D}+\frac D2+O(D^2)
 \qquad(D\downarrow0).
\end{equation}
\end{corollary}

\begin{proof}
For $D<q_*$ no nonzero rank vertex $D_{\mathbf k}$ is reachable: the first
completed saturation of any nontrivial block costs $q_i\ge q_*$. Thus the
only relevant edges emanate from the all-pure vertex. On such an edge the
cost is
\[
 q_i h_2\!\left(\frac{D}{2q_i}\right).
\]
For fixed $D>0$ this quantity is increasing in $q_i$, because
\[
 \frac{\partial}{\partial q}
 \left[qh_2\!\left(\frac{D}{2q}\right)\right]
 =-\log\left(1-\frac{D}{2q}\right)>0.
\]
The minimum is therefore attained at $q_*$. The endpoint follows by
continuity, and the expansion of $h_2$ gives~\eqref{eq:small-D-asymptotic}.
\end{proof}

\begin{corollary}[Entropy-gap certificate]
\label{cor:global-certificate}
The function $\Phi_{q,\mathbf n}$ is continuous and strictly increasing on
$[0,D_{\max}]$. Every state in $\mathcal C_q$ therefore satisfies
\begin{equation}
\label{eq:global-inverse}
 D\le \Phi_{q,\mathbf n}^{-1}\bigl(S(\rho)-S_{\min}\bigr).
\end{equation}
The inverse bound therefore holds on the full interval $[0,D_{\max}]$.
\end{corollary}

\begin{proof}
Let $D_m\to D$ and choose a minimizing vector $d^{(m)}$ for each $D_m$.
Compactness of the product box gives a convergent subsequence
$d^{(m_\ell)}\to d$.  The limit is feasible at $D$, and continuity of the
one-block profiles gives
\[
 \Phi_{q,\mathbf n}(D)
 \le \sum_i q_i\beta_{n_i}(d_i)
 =\lim_{\ell\to\infty}\Phi_{q,\mathbf n}(D_{m_\ell}),
\]
which proves the lower-semicontinuity inequality.  Conversely, fix a
minimizer $d$ at $D$.  For all sufficiently large $m$, one can change positive
coordinates if $D_m<D$, or coordinates with unused capacity if $D_m>D$, by
a total amount tending to zero so that
$\sum_iq_i d_i^{(m)}=D_m$ and $d^{(m)}\to d$.  Therefore
\[
 \limsup_{m\to\infty}\Phi_{q,\mathbf n}(D_m)
 \le\lim_{m\to\infty}\sum_iq_i\beta_{n_i}(d_i^{(m)})
 =\Phi_{q,\mathbf n}(D).
\]
Thus $\Phi_{q,\mathbf n}$ is continuous.

For strict monotonicity, let $D_2>D_1$ and choose a minimizing vector $(d_i)$
at $D_2$. Reduce one or more positive coordinates continuously until the
weighted sum equals $D_1$. Every $\beta_{n_i}$ is strictly increasing on its
domain, so this produces a strictly smaller objective value. Hence
$\Phi_{q,\mathbf n}(D_1)<\Phi_{q,\mathbf n}(D_2)$.
The inverse certificate follows from Theorem~\ref{thm:global-envelope}.
\end{proof}

\begin{proposition}[Two-qubit envelope and optimizer switching]
\label{prop:two-qubit-switch}
Let $n_1=n_2=2$, $q_1=a$, and $q_2=b=1-a$ with $0<a<b$. There is a unique
$D_c\in(a,b)$ satisfying
\begin{equation}
\label{eq:two-qubit-crossing}
 a\log2+b h_2\!\left(\frac{D_c-a}{2b}\right)
 =b h_2\!\left(\frac{D_c}{2b}\right).
\end{equation}
The exact envelope is
\begin{equation}
\label{eq:two-qubit-envelope}
 \Phi(D)=
 \begin{cases}
 a h_2\!\left(\dfrac{D}{2a}\right),&0\le D\le a,\\[2mm]
 a\log2+b h_2\!\left(\dfrac{D-a}{2b}\right),&a\le D\le D_c,\\[2mm]
 b h_2\!\left(\dfrac{D}{2b}\right),&D_c\le D\le b,\\[2mm]
 b\log2+a h_2\!\left(\dfrac{D-b}{2a}\right),&b\le D\le1.
 \end{cases}
\end{equation}
At $D_c$ the minimizing edge switches between two nonnested allocations:
just below the transition one may take
$(d_1,d_2)=(1,(D-a)/b)$, whereas just above it one may take
$(d_1,d_2)=(0,D/b)$. Thus optimal rank allocations need not be nested as the
total distance increases.
\end{proposition}

\begin{proof}
For two qubits the objective is
$a h_2(d_1/2)+b h_2(d_2/2)$ under $ad_1+bd_2=D$ and
$0\le d_1,d_2\le1$. It is concave on each feasible line segment, so a
minimum is attained at an endpoint. For $0\le D\le a$, the two endpoint
costs are $a h_2(D/(2a))$ and $b h_2(D/(2b))$; the first is smaller because
$q\mapsto qh_2(D/(2q))$ is strictly increasing.

For $a\le D\le b$, the endpoint costs are
\[
 B(D):=a\log2+b h_2\!\left(\frac{D-a}{2b}\right),
 \qquad
 C(D):=b h_2\!\left(\frac{D}{2b}\right).
\]
Their difference $F=B-C$ satisfies
\[
 F'(D)=\frac12\left[
 h_2'\!\left(\frac{D-a}{2b}\right)
 -h_2'\!\left(\frac{D}{2b}\right)\right]>0,
\]
since $h_2'$ is strictly decreasing on $(0,1/2)$. Also $F(a)<0$. To see
$F(b)>0$, put $c=2b-1\in(0,1)$. A direct simplification gives
\[
 F(b)=\frac12\bigl[(1+c)\log(1+c)-c\log c-2c\log2\bigr].
\]
This expression vanishes at $c=0,1$ and has derivative
$\frac12\log((1+c)/(4c))$.  The derivative is positive on $(0,1/3)$ and
negative on $(1/3,1)$, so the function rises from zero and then decreases
back to zero.  It is therefore strictly positive for $0<c<1$.  Hence there
is a unique crossing $D_c\in(a,b)$, proving the two middle branches.

For $b\le D\le1$, compare $B(D)$ with
\[
 E(D):=b\log2+a h_2\!\left(\frac{D-b}{2a}\right).
\]
One has $B(b)>E(b)$ and $B(1)=E(1)=\log2$. Moreover, for $D<1$,
\[
 \frac{D-a}{2b}>\frac{D-b}{2a},
\]
so the monotonic decrease of $h_2'$ gives $(B-E)'<0$. Therefore
$E(D)<B(D)$ on $(b,1)$, which gives the last branch.
\end{proof}

\begin{example}[A nonnested switching transition]
For $(q_1,q_2)=(0.2,0.8)$, equation~\eqref{eq:two-qubit-crossing} gives
\[
 D_c\simeq0.4239432156.
\]
The optimizer first mixes the lighter block, then keeps it maximally mixed
while beginning to mix the heavier block. At $D_c$ the optimal edge changes:
the lighter block returns to a pure endpoint and the heavier block alone
carries the distance until $D=0.8$. The minimizing rank allocation is therefore not, in general, nested as the
distance increases.
\end{example}

The smallest nontrivial weight $q_*$ controls the first branch. Beyond it,
the active edge depends on the full profile $(q_i,n_i)$.  The switching
threshold in Proposition~\ref{prop:two-qubit-switch} satisfies
$D_c\to2/5$ as $a\downarrow0$ and $D_c\to1/2$ as $a\uparrow1/2$; the first
limit is also the point where the binary entropy arc has its largest gap from
its chord in Lemma~\ref{lem:one-block-convexification} below.

\section{Convexification and a certified surrogate}
\label{sec:convexification}

The exact envelope in Theorem~\ref{thm:rank-lattice} is nonconvex because
the one-block curves have the upward kinks in~\eqref{eq:beta-kink}.  Its
lower convex envelope nevertheless has a simpler structure that retains the
fixed weight and dimension profiles.

For $n\ge1$, let $\beta_n^{\mathrm{cvx}}$ denote the greatest convex function
below $\beta_n$ on $[0,\delta_n]$.  For $1\le k<n$ define
\begin{equation}
\label{eq:stage-slope}
 s_k:=\frac{\log(k+1)-\log k}{\delta_{k+1}-\delta_k}
 =\frac{k(k+1)}{2}\log\!\left(1+\frac1k\right).
\end{equation}

\begin{lemma}[One-block convexification]
\label{lem:one-block-convexification}
The function $\beta_n^{\mathrm{cvx}}$ is the polygonal interpolation through
the rank points $(\delta_k,\log k)$, $1\le k\le n$.  On the $k$th stage its
slope is $s_k$, and the sequence $(s_k)$ is strictly increasing.

Set
\begin{equation}
\label{eq:Delta-k}
 c_k:=\left(1+\frac1k\right)^{k+1},
 \qquad
 \Delta_k:=\log\!\left(1+\frac{1}{k c_k}\right).
\end{equation}
Then
\begin{equation}
\label{eq:arc-chord-gap}
 \max_{\delta_k\le d\le\delta_{k+1}}
 \bigl[\beta_n(d)-\beta_n^{\mathrm{cvx}}(d)\bigr]
 =\Delta_k.
\end{equation}
The maximum occurs at
\[
 p=\frac{c_k}{1+k c_k},
 \qquad
 r=\frac{1}{1+k c_k}.
\]
Moreover, $\Delta_k$ is strictly decreasing,
\[
 \Delta_1=\log\frac54,
 \qquad
 \Delta_k=\frac{1}{e k}+O(k^{-2}).
\]
\end{lemma}

\begin{proof}
Each stage of $\beta_n$ is concave and therefore lies above the chord joining
its endpoints.  The chord slope is~\eqref{eq:stage-slope}.  The continuous
function
\[
 x\longmapsto x(x+1)\log(1+1/x)
\]
has derivative
$(2x+1)\log(1+1/x)-1>0$, since
$\log(1+1/x)>1/(x+1)$.  Hence the chord slopes increase, and the chords form
the greatest convex minorant.

On stage $k$, the gap between $\beta_n$ and its chord is strictly concave and
vanishes at the endpoints.  Its maximum is characterized by
$\beta_n'(d)=s_k$, or $p/r=c_k$.  Together with $kp+r=1$ this gives the
displayed $p$ and $r$, as well as
\[
 d-\delta_k=\frac{2}{k(1+k c_k)}.
\]
Substitution into the entropy and into the chord yields
\[
 \beta_n(d)-\beta_n^{\mathrm{cvx}}(d)
 =\log\!\left(1+\frac{1}{k c_k}\right).
\]
Finally,
\[
 \frac{d}{dx}\left[\log x+(x+1)\log(1+1/x)\right]
 =\log(1+1/x)>0,
\]
so $k c_k$ is strictly increasing and $\Delta_k$ is strictly decreasing.
The remaining statements follow from $c_1=4$ and the standard expansion
$(1+1/k)^{k+1}=e(1+O(k^{-1}))$.
\end{proof}

Let $\Phi_{q,\mathbf n}^{\mathrm{cvx}}$ denote the lower convex envelope of
$\Phi_{q,\mathbf n}$ on $[0,D_{\max}]$, and set
$M:=\max_i n_i$.  For $\ell=0,\ldots,M-1$ define the synchronized rank
vertices
\begin{equation}
\label{eq:synchronized-vertices}
 \overline D_\ell
 :=\sum_iq_i\delta_{\min(n_i,\ell+1)},
 \qquad
 \overline E_\ell
 :=\sum_iq_i\log\min(n_i,\ell+1).
\end{equation}
Thus $\ell=0$ is the all-pure vertex and $\ell=M-1$ is the endpoint
$D_{\max}$.

\begin{theorem}[Convexification and sharp nonconvexity gap]
\label{thm:convexification}
The lower convex envelope $\Phi_{q,\mathbf n}^{\mathrm{cvx}}$ is the
polygonal interpolation through the synchronized vertices
$(\overline D_\ell,\overline E_\ell)$ in~\eqref{eq:synchronized-vertices}.
Equivalently, these synchronized vertices form the lower convex hull of the
full rank-vertex cloud $\{(D_{\mathbf k},E_{\mathbf k})\}$, and
\begin{equation}
\label{eq:convexified-primal}
 \Phi_{q,\mathbf n}^{\mathrm{cvx}}(D)
 =
 \min\left\{
 \sum_iq_i\beta_{n_i}^{\mathrm{cvx}}(d_i):
 0\le d_i\le\delta_{n_i},\
 \sum_iq_id_i=D
 \right\},
\end{equation}
and it has the dual representation
\begin{equation}
\label{eq:convexified-dual}
 \Phi_{q,\mathbf n}^{\mathrm{cvx}}(D)
 =
 \sup_{\lambda\in\mathbb R}
 \left\{
 \lambda D+
 \sum_iq_i
 \min_{1\le k\le n_i}
 \bigl(\log k-\lambda\delta_k\bigr)
 \right\}.
\end{equation}

For $1\le k<M$ and
$D\in[\overline D_{k-1},\overline D_k]$, define
\[
 q_{\diamond,k}:=\max\{q_i:n_i\ge k+1\}.
\]
Then
\begin{equation}
\label{eq:convexification-gap}
 0\le
 \Phi_{q,\mathbf n}(D)-\Phi_{q,\mathbf n}^{\mathrm{cvx}}(D)
 \le q_{\diamond,k}\Delta_k.
\end{equation}
In particular, with
$q_\diamond:=\max\{q_i:n_i\ge2\}$,
\begin{equation}
\label{eq:global-convexification-gap}
 0\le
 \Phi_{q,\mathbf n}(D)-\Phi_{q,\mathbf n}^{\mathrm{cvx}}(D)
 \le q_\diamond\log\frac54
\end{equation}
for all $D$.  The constant $\log(5/4)$ is sharp.
\end{theorem}

\begin{proof}
Introduce the budget variables $z_i=q_id_i$ and the extended functions
\[
 f_i(z):=
 \begin{cases}
 q_i\beta_{n_i}(z/q_i),&0\le z\le q_i\delta_{n_i},\\
 +\infty,&\text{otherwise}.
 \end{cases}
\]
Then~\eqref{eq:Phi-variational} is the infimal convolution
$\Phi=f_1\mathbin{\square}\cdots\mathbin{\square}f_r$.  Hence, with the
standard convex conjugate,
\[
 \Phi^*(\lambda)=\sum_i f_i^*(\lambda)
 =\sum_iq_i\sup_{0\le d\le\delta_{n_i}}
 \{\lambda d-\beta_{n_i}(d)\}.
\]
Replacing $f_i$ by its lower closed convex envelope leaves $f_i^*$ unchanged.
Fenchel--Moreau duality~\cite{Rockafellar1970} therefore identifies the
biconjugate $\Phi^{**}$ with the infimal convolution in
\eqref{eq:convexified-primal}; on the compact interval this biconjugate is
precisely the lower convex envelope
$\Phi_{q,\mathbf n}^{\mathrm{cvx}}$.  The same calculation gives
\eqref{eq:convexified-dual}.  Indeed, on each rank stage
$\beta_n(d)-\lambda d$ is concave, so its minimum is attained at a rank
endpoint and equals $\min_k(\log k-\lambda\delta_k)$.

The slopes of all $k$th one-block chord segments are the same number $s_k$,
independent of $i$, and $s_k$ increases with $k$.  The convexified allocation
therefore fills every available stage of index below $k$ before using stage
$k$.  If
\[
 Q_k:=\sum_{i:n_i\ge k+1}q_i,
\]
then
\[
 \overline D_k-\overline D_{k-1}
 =Q_k\frac{2}{k(k+1)},
 \qquad
 \overline E_k-\overline E_{k-1}
 =Q_k\log\!\left(1+\frac1k\right),
\]
whose ratio is $s_k$.  This proves the synchronized-vertex representation.

Fix now $D\in[\overline D_{k-1},\overline D_k]$.  In the convexified problem,
all stages below $k$ are filled and only the $k$th-stage capacities need to
share the remaining budget.  Since their slopes agree, an optimal allocation
may be chosen with every coordinate at a rank point except possibly one
coordinate $j$ in stage $k$.  This allocation is feasible for the original
problem, and therefore
\[
 \Phi_{q,\mathbf n}(D)
 \le
 \Phi_{q,\mathbf n}^{\mathrm{cvx}}(D)
 +q_j\bigl[
 \beta_{n_j}(d_j)-\beta_{n_j}^{\mathrm{cvx}}(d_j)
 \bigr].
\]
Lemma~\ref{lem:one-block-convexification} gives
\eqref{eq:convexification-gap}, and
$q_{\diamond,k}\Delta_k\le q_\diamond\log(5/4)$ gives
\eqref{eq:global-convexification-gap}.  Sharpness follows from a single
nontrivial two-dimensional block, or one such block together with
one-dimensional blocks, evaluated at $d=2/5$.
\end{proof}

The convex surrogate uses only the synchronized rank points rather than the
full product rank lattice.  For a pointwise evaluation, the interval
containing $D$ can be located from~\eqref{eq:synchronized-vertices}, after
which linear interpolation gives $\Phi_{q,\mathbf n}^{\mathrm{cvx}}(D)$.
Theorem~\ref{thm:convexification} then supplies an explicit absolute error
certificate for the exact nonconvex value.

\section{Entropy curvature on a fixed support}
\label{sec:fixed-support}

On a fixed support $P$, a faithful state $\sigma$ is an interior point of the
state space on $P\mathcal H$. There the entropy is smooth and strictly
concave. Thus its second differential is negative definite on nonzero
traceless Hermitian directions. Fixed-support quadratic growth cannot be
attributed to positive entropy curvature; when it occurs, it must come from
the feasible geometry.

\begin{proposition}[No zero-slope minimum along a feasible segment]
\label{prop:no-zero-slope}
Let $\mathcal C$ be convex and let $\sigma\in\mathcal C$ be faithful on a
fixed support $P$. Suppose $\tau\in\mathcal C$, $\tau\ne\sigma$, has the same
support and set $X=\tau-\sigma$. If
\[
 \left.\frac{d}{dt}S(\sigma+tX)\right|_{t=0}=0,
\]
then $\sigma$ is not a local entropy minimizer on $\mathcal C$.
\end{proposition}

\begin{proof}
The segment $\sigma+tX$ is feasible for $0\le t\le1$. The function
$g(t)=S(\sigma+tX)$ is strictly concave and differentiable near zero. For a
differentiable concave function, the graph lies below its tangent line, so for
$t>0$,
\[
 g(t)<g(0)+tg'(0)=g(0).
\]
Thus feasible points arbitrarily close to $\sigma$ have smaller entropy.
\end{proof}

\section{Polyhedral fixed-support marginal constraints give linear growth}
\label{sec:polyhedral}

For the marginal problem, fix $q\in P_*$ and let $I=\supp(q)$. Assume that the support is locally fixed: there exists a
neighbourhood $U$ of $q$ such that
\begin{equation}
\label{eq:fixed-support}
 p\in\Pi\cap U\quad\Longrightarrow\quad\supp(p)=I.
\end{equation}
On the relative interior of $\Delta_I$, the functional $F$ is $C^2$. Let
$K_q$ be the feasible tangent cone of $\Pi$ at $q$, restricted to the affine
span of $\Delta_I$:
\[
 K_q:=\{v:\exists\epsilon>0\text{ such that }q+tv\in\Pi
 \text{ for all }0\le t\le\epsilon\}.
\]

\begin{theorem}[First-order coefficient at a polyhedral marginal minimizer]
\label{thm:polyhedral-linear}
Assume that $\Pi$ is a polytope and that~\eqref{eq:fixed-support} holds. Let
$\nabla_I F(q)$ denote the gradient of $F$ relative to the fixed-support face
$\Delta_I$. If $K_q=\{0\}$, then $q$ is locally isolated in $\Pi$.
Otherwise define
\begin{equation}
\label{eq:kappa}
 \kappa_q:=\min\{\nabla_I F(q)\cdot v:
 v\in K_q,\ \|v\|_1=1\}.
\end{equation}
Then $\kappa_q>0$ and
\begin{equation}
\label{eq:liminf-kappa}
 \liminf_{\substack{p\to q\\p\in\Pi,\ p\ne q}}
 \frac{F(p)-F(q)}{\|p-q\|_1}=\kappa_q.
\end{equation}
In particular, for every $0<c<\kappa_q$ there is a neighbourhood $U_q$ such
that
\begin{equation}
\label{eq:linear-growth}
 F(p)-F(q)\ge c\|p-q\|_1,
 \qquad p\in\Pi\cap U_q.
\end{equation}
\end{theorem}

\begin{proof}
Assume $K_q\ne\{0\}$. The normalized section
$K_q\cap\{\|v\|_1=1\}$ is compact. If a nonzero feasible tangent direction
$v$ satisfied $\nabla_I F(q)\cdot v=0$, then strict concavity would give, for
all sufficiently small $t>0$,
\[
 F(q+tv)<F(q)+t\nabla_I F(q)\cdot v=F(q),
\]
contradicting minimality. Thus the directional derivative is positive on the
normalized tangent cone, and compactness gives $\kappa_q>0$.

For any sequence $p_n\to q$ with $p_n\in\Pi$ and $p_n\ne q$, write
$v_n=(p_n-q)/\|p_n-q\|_1$. Convexity implies $v_n\in K_q$, and the relative
Taylor expansion on $\Delta_I$ gives
\[
 \frac{F(p_n)-F(q)}{\|p_n-q\|_1}
 =\nabla_I F(q)\cdot v_n+O(\|p_n-q\|_1).
\]
Hence the liminf is at least $\kappa_q$. Conversely, the minimum in
\eqref{eq:kappa} is attained at a normalized $v_*$, and the feasible ray
$q+tv_*$ gives a quotient converging to $\kappa_q$.
\end{proof}

\begin{corollary}[No quadratic leading law on a fixed-support marginal polytope]
Under the hypotheses of Theorem~\ref{thm:polyhedral-linear}, a nontrivial
family $p_t\to q$ cannot satisfy
\[
 F(p_t)-F(q)\asymp\|p_t-q\|_1^2.
\]
The local growth has a positive first-order coefficient.
\end{corollary}

\section{Support opening produces a logarithmic boundary term}
\label{sec:support-opening}

A different mechanism appears when an admissible path creates weight in a
previously empty sector.

\begin{proposition}[Support-opening asymptotics for the marginal functional]
\label{prop:support-opening}
Let $q\in\Delta_r$ and let $s\in\Delta_r$ be distinct. Set
$p_t=(1-t)q+ts$ and
\[
 \zeta:=\sum_{j\notin\supp(q)}s_j.
\]
If $\zeta>0$, then for $F(p)=H(p)+m\cdot p$,
\begin{equation}
\label{eq:support-opening}
 F(p_t)-F(q)=\zeta\,t\log\frac1t+O(t)
 \qquad(t\downarrow0).
\end{equation}
In particular, the ratio $[F(p_t)-F(q)]/t$ diverges to $+\infty$ whenever $q$
is a minimizer along the feasible ray.
\end{proposition}

\begin{proof}
For $j\notin\supp(q)$ one has $(p_t)_j=ts_j$, and hence
\[
 -(p_t)_j\log(p_t)_j
 =ts_j\log\frac1t-ts_j\log s_j.
\]
Summing these terms gives $\zeta t\log(1/t)+O(t)$. The coordinates in
$\supp(q)$ vary smoothly with $t$, and the affine term $m\cdot p_t$
contributes only $O(t)$.
\end{proof}

The same singularity occurs when a density matrix opens a zero eigenvalue.
The first stage of the one-block profile in
Theorem~\ref{thm:one-block-profile} makes this boundary mechanism quantitative.
Here the rank opening is already separated by the assumed block decomposition;
the companion analysis~\cite{NasreddineRankStratified2026} treats the
non-block-diagonal case, where coherent support motion must first be separated
from genuine rank creation.

\section{Quadratic and higher-order growth from curved active constraints}
\label{sec:curved}

Quadratic entropy stability can occur on a smooth constrained manifold even
though entropy itself is concave. The relevant mechanism is curvature of the
active constraint. We use the standard finite-dimensional second-order
sufficient condition in a local coordinate chart of a fixed-support stratum.

Let $x\in\mathbb R^n$ be local coordinates for trace-one Hermitian matrices on
a fixed support, and write $s(x)$ for entropy. Consider
\begin{equation}
\label{eq:nonlinear-constraints}
 \mathcal X=\{x:g_a(x)\le0\ (a\in A),\ h_b(x)=0\ (b\in B)\},
\end{equation}
where all functions are $C^2$ near a feasible point $x_*$. Let $A_*$ be the
active inequality indices and assume a standard constraint qualification,
for example LICQ. For KKT multipliers define the critical cone
\begin{equation}
\label{eq:critical-cone}
 \mathcal C_*:=\left\{v:
 \begin{array}{ll}
 \nabla h_b(x_*)\cdot v=0,&b\in B,\\
 \nabla g_a(x_*)\cdot v=0,&a\in A_*,\ \lambda_a>0,\\
 \nabla g_a(x_*)\cdot v\le0,&a\in A_*,\ \lambda_a=0
 \end{array}\right\}.
\end{equation}

\begin{theorem}[Quadratic growth under a constrained second-order condition]
\label{thm:quadratic}
Suppose there are KKT multipliers $\lambda_a\ge0$ and $\mu_b$ such that
\begin{equation}
\label{eq:KKT}
 \nabla s(x_*)+\sum_{a\in A_*}\lambda_a\nabla g_a(x_*)
 +\sum_{b\in B}\mu_b\nabla h_b(x_*)=0,
\end{equation}
with complementarity for the active inequalities. Let
\[
 L(x)=s(x)+\sum_{a\in A_*}\lambda_ag_a(x)
 +\sum_{b\in B}\mu_bh_b(x).
\]
If there exists $\eta>0$ such that
\[
 v^T\nabla^2L(x_*)v\ge\eta\|v\|_2^2
\]
for every nonzero $v\in\mathcal C_*$, then $x_*$ is a strict local minimizer
and there exist $c,r>0$ such that
\begin{equation}
\label{eq:quadratic-growth}
 s(x)-s(x_*)\ge c\|x-x_*\|_2^2
 \qquad(x\in\mathcal X,\ \|x-x_*\|_2<r).
\end{equation}
\end{theorem}

\begin{proof}
This is the standard second-order sufficient condition for nonlinear
programming under LICQ, applied to the smooth entropy restricted to a faithful
support stratum; see, for example, Bonnans--Shapiro~\cite{BonnansShapiro2000}.
\end{proof}

On a faithful support stratum, $\nabla^2s(x_*)$ is negative definite on
nonzero trace-zero directions. Thus the positive quadratic form in
Theorem~\ref{thm:quadratic} must come from the Hessians of the active
constraints. For polyhedral constraints those Hessians vanish.

\begin{example}[A quadratic minimum created by a curved constraint]
Consider the interior point $q=(1/2,1/3,1/6)$ of the three-outcome simplex.
In the affine plane $p_1+p_2+p_3=1$, choose a nonzero direction $v$ tangent to
the entropy level set at $q$, so $\nabla H(q)\cdot v=0$, and a tangent
direction $n$ with $\nabla H(q)\cdot n>0$. Choose
\[
 a>-\frac{v^T\nabla^2H(q)v}{2\,\nabla H(q)\cdot n},
\]
and then choose $R>0$ small enough that $aR\le1$ and the convex set
\[
 K=\{q+xv+yn:|x|\le R,\ 0\le y\le R,\ y\ge ax^2\}
\]
is contained in the simplex. Along its lower boundary,
\[
 H(q+xv+ax^2n)-H(q)
 =x^2\left(a\nabla H(q)\cdot n+\frac12v^T\nabla^2H(q)v\right)+o(x^2).
\]
The Hessian term is negative, so a sufficiently large $a$ makes the
coefficient positive. Shrinking the neighbourhood then gives a positive
quadratic lower bound throughout the local epigraph.
\end{example}

\paragraph{Higher-order degeneracy}
The linear, quadratic, and support-opening mechanisms above do not form an
exhaustive trichotomy. To see this already for classical entropy, choose an
interior point $q$, a tangent direction $v$ with $\nabla H(q)\cdot v=0$, and
a direction $n$ with $b:=\nabla H(q)\cdot n>0$. By the implicit-function
theorem there is a smooth local function $\varphi$ with
\[
 H(q+xv+\varphi(x)n)=H(q),\qquad \varphi(0)=\varphi'(0)=0.
\]
Because $D^2H(q)[v,v]<0$, one has $\varphi''(0)>0$, so $\varphi$ is locally
convex. For $a>0$ the epigraph constraint
\[
 y\ge\varphi(x)+ax^4
\]
is locally convex, and along its boundary
\[
 H(q+xv+[\varphi(x)+ax^4]n)-H(q)=ab\,x^4+o(x^4).
\]
Thus quartic growth occurs when the second-order constrained term vanishes.
For smooth constraints the same construction gives higher even orders.
At the $C^2$ regularity of Theorem~\ref{thm:quadratic}, noninteger powers
$|x|^\alpha$ with $\alpha>2$ can also be realized by replacing $x^4$ with
$|x|^\alpha$.

\section{Composition of stability moduli}
\label{sec:composition}

The gap identity~\eqref{eq:exact-gap} allows marginal and conditional moduli
to be combined without forcing them into a common power law. Fix a minimizer
$\sigma=\bigoplus_iq_i\sigma_i$ with $q\in P_*$ and
$\sigma_i\in\mathcal M_i$ for $q_i>0$. Set
\[
 d_0(p):=\|p-q\|_1,
 \qquad
 d_i(\rho_i):=\dist_1(\rho_i,\mathcal M_i).
\]
Suppose there are radii $\eta_0>0$ and $\eta_i>0$ for the active blocks,
and nondecreasing functions $\omega_0,\omega_i$ vanishing only at zero, such
that
\begin{align}
 F(p)-F_*&\ge\omega_0(d_0(p))&&\text{when }d_0(p)<\eta_0,
 \label{eq:marginal-modulus}\\
 S(\rho_i)-m_i&\ge\omega_i(d_i(\rho_i))&&\text{when }d_i(\rho_i)<\eta_i.
 \label{eq:conditional-modulus}
\end{align}

\begin{theorem}[Composition of stability moduli]
\label{thm:composition}
After restricting to a sufficiently small neighbourhood of $\sigma$, one has
$p_i\ge q_i/2$ on every active block, all local assumptions
$d_i(\rho_i)<\eta_i$ hold, and therefore
\begin{equation}
\label{eq:composed-gap}
 S(\rho)-S_{\min}
 \ge\omega_0(d_0)+\frac12\sum_{q_i>0}q_i\omega_i(d_i).
\end{equation}
Moreover,
\begin{equation}
\label{eq:distance-budget}
 \dist_1(\rho,\mathcal M)
 \le d_0+\sum_{q_i>0}q_id_i,
\end{equation}
where $\mathcal M$ is the global minimizer set. Consequently a local modulus
in terms of total distance $D$ is the infimal convolution
\begin{equation}
\label{eq:inf-convolution}
 \Omega_\sigma(D):=
 \inf\left\{
 \omega_0(x_0)+\frac12\sum_{q_i>0}q_i\omega_i(x_i):
 x_0+\sum_{q_i>0}q_ix_i\ge D
 \right\},
\end{equation}
with the variables restricted to their local domains. Then
\[
 S(\rho)-S_{\min}\ge
 \Omega_\sigma\bigl(\dist_1(\rho,\mathcal M)\bigr)
\]
locally near $\sigma$.
\end{theorem}

\begin{proof}
Since $p\to q$, one has $p_i\ge q_i/2$ locally for every $q_i>0$. Local
proximity of the full block-diagonal state also implies proximity of each
active normalized block. Indeed,
\[
 q_i\|\rho_i-\sigma_i\|_1
 \le \|p_i\rho_i-q_i\sigma_i\|_1+|p_i-q_i|,
\]
so $\rho\to\sigma$ and $p\to q$ imply
$\|\rho_i-\sigma_i\|_1\to0$. Hence, after shrinking the neighbourhood,
$d_i(\rho_i)<\eta_i$ for every active block.

The gap identity gives
\[
 S(\rho)-S_{\min}
 \ge\omega_0(d_0)+\sum_{q_i>0}p_i\omega_i(d_i),
\]
and $p_i\ge q_i/2$ yields~\eqref{eq:composed-gap}. For each active block
choose $\widetilde\sigma_i\in\mathcal M_i$ realizing $d_i$ and set
$\widetilde\sigma=\bigoplus_iq_i\widetilde\sigma_i$. Since
$q\in P_*$, this is a global minimizer. Lemma~\ref{lem:block-distance} gives
\eqref{eq:distance-budget}. Minimizing the right-hand side of
\eqref{eq:composed-gap} under this distance budget gives
\eqref{eq:inf-convolution}.
\end{proof}

\begin{corollary}[Power-law bookkeeping]
If locally $\omega_j(t)\ge a_jt^{\alpha_j}$ with $\alpha_j\ge1$, then there is
$C>0$ such that
\[
 \Omega_\sigma(D)\ge CD^\alpha,
 \qquad \alpha=\max_j\alpha_j,
\]
for all sufficiently small $D$.
\end{corollary}

\begin{proof}
Shrink the local domains so that all variables lie in $[0,1]$. Since
$\alpha_j\le\alpha$, one has $x^{\alpha_j}\ge x^\alpha$ there. Absorbing
positive weights into constants gives a lower bound by a positive multiple of
a finite weighted sum of $x_j^\alpha$. The distance budget and convexity of
$t\mapsto t^\alpha$ then give the stated bound.
\end{proof}

\section{Symmetry sectors and symmetry-resolved entropy}
\label{sec:symmetry}

\subsection{Charge-sector decomposition}
Suppose a reduced state commutes with a conserved charge $Q_A$. Writing
$P_s$ for the spectral projector of $Q_A$ gives
\begin{equation}
\label{eq:charge-decomposition}
 \rho_A=\bigoplus_s p(s)\rho_A(s),
 \qquad p(s)=\Tr(P_s\rho_A),
\end{equation}
with normalized sector states whenever $p(s)>0$. Its entropy is
\begin{equation}
\label{eq:sector-entropy}
 S(\rho_A)=H(p)+\sum_s p(s)S(\rho_A(s)).
\end{equation}
This decomposition is standard in the quantum-information treatment of
superselection rules and in symmetry-resolved entanglement
\cite{BartlettRudolphSpekkens2007,GoldsteinSela2018,
BonsignoriRuggieroCalabrese2019,Jones2022}. The Shannon term describes charge
or fluctuation uncertainty, while the weighted sector entropies describe
within-sector or configurational mixedness. Majorization has also been used
to derive model-dependent bounds on number and configurational entropy in
symmetry-protected phases~\cite{MonkmanSirker2023}.

A recent symmetry-resolved trace-distance construction for thermalizing
$SU(2)$ systems separates a probability contribution from a configurational
intra-sector contribution~\cite{ShenChenWang2026}. The decomposition is
conceptually parallel to~\eqref{eq:sector-entropy}, but the present question is
different: given a weighted within-sector entropy, how close must the state be
to sectorwise purity?

\subsection{Sector-purity certificate at fixed charge statistics}
When the charge statistics $p(s)$ are fixed, restrict attention to the occupied
sectors $p(s)>0$. Then $H(p)$ is constant and the entropy excess above the
blockwise-pure minimum is
\begin{equation}
\label{eq:sector-excess}
 \Delta S_{\mathrm{sector}}
 :=S(\rho_A)-H(p)
 =\sum_s p(s)S(\rho_A(s)).
\end{equation}
Let $n_s=\dim P_s\mathcal H_A$ in the finite-dimensional or regulated model,
and define
\[
 D_{\mathrm{sector}}
 :=\dist_1\!\left(
 \rho_A,
 \left\{\bigoplus_sp(s)P_s^{(1)}:
 P_s^{(1)}\ \text{rank one in sector }s\right\}
 \right).
\]
Theorem~\ref{thm:global-envelope} gives
\begin{equation}
\label{eq:sector-global}
 \Delta S_{\mathrm{sector}}
 \ge\Phi_{p,\mathbf n}(D_{\mathrm{sector}}),
 \qquad
 D_{\mathrm{sector}}
 \le\Phi_{p,\mathbf n}^{-1}(\Delta S_{\mathrm{sector}}).
\end{equation}
The bound is attained throughout the admissible interval.  Since
$\Phi_{p,\mathbf n}^{\mathrm{cvx}}\le\Phi_{p,\mathbf n}$ and the convex
surrogate is increasing, one also has the coarser certified inversion
\[
 D_{\mathrm{sector}}
 \le
 \bigl(\Phi_{p,\mathbf n}^{\mathrm{cvx}}\bigr)^{-1}
 (\Delta S_{\mathrm{sector}}).
\]
At small distance the exact bound reduces to
Corollary~\ref{cor:binary-first-branch}; beyond that branch, the remaining
sector weights and dimensions determine the envelope.

\begin{example}[A binomial excitation-sector profile]
\label{ex:binomial-sectors}
Consider $L$ qubits and total excitation number
$N=\sum_{\ell=1}^L |1\rangle\!\langle1|_\ell$. Projective dephasing of
$|+\rangle^{\otimes L}$ in the eigenspaces of $N$ produces the fixed sector
weights and sector dimensions
\[
 p(s)=2^{-L}\binom{L}{s},
 \qquad
 n_s=\binom{L}{s},
 \qquad s=0,\ldots,L.
\]
The endpoint sectors are one dimensional. Among the nontrivial sectors,
\[
 p_*=\frac{L}{2^L},
\]
whereas the full admissible distance range is
\begin{equation}
\label{eq:binomial-Dmax}
 D_{\max}
 =2\sum_{s=0}^L p(s)\left(1-\frac1{n_s}\right)
 =2\left(1-\frac{L+1}{2^L}\right).
\end{equation}
Thus the first binary branch becomes exponentially narrow while the full
distance interval approaches $2$. For example, $L=8$ gives $p_*=1/32$ and
$D_{\max}=247/128\simeq1.9297$, whereas $L=16$ gives $p_*=1/4096$ and
$D_{\max}\simeq1.99948$.

The convexification bound has a different many-sector scale.  Here
\[
 q_\diamond=\max_s p(s)
 =2^{-L}\binom{L}{\lfloor L/2\rfloor}
 \sim\sqrt{\frac{2}{\pi L}},
\]
so Theorem~\ref{thm:convexification} gives the uniform estimate
\[
 \sup_{0\le D\le D_{\max}}
 \bigl[\Phi_{p,\mathbf n}(D)-\Phi_{p,\mathbf n}^{\mathrm{cvx}}(D)\bigr]
 \le
 \log\frac54\sqrt{\frac{2}{\pi L}}\,[1+o(1)].
\]
Thus the fixed-weight convex surrogate approaches the exact envelope uniformly
at order $L^{-1/2}$ even though the first exact binary branch shrinks
exponentially.  The dephased product state itself has pure conditional
sectors; the next example adds within-sector mixedness without changing the
charge statistics.
\end{example}

\begin{example}[Phase damping at fixed charge statistics]
\label{ex:three-qubit-dephasing}
A number-preserving noise process can create mixedness inside these sectors. Take $L=3$ and let
$\mathcal E_\lambda$ be the single-qubit phase-damping channel
\[
 \mathcal E_\lambda(\tau)
 =\frac{1+\lambda}{2}\tau
 +\frac{1-\lambda}{2}Z\tau Z,
 \qquad 0\le\lambda\le1.
\]
It preserves computational-basis populations and multiplies a one-qubit
coherence by $\lambda$. Apply
$\mathcal E_\lambda^{\otimes3}$ to $|+\rangle^{\otimes3}$ and then dephase
with respect to total excitation number. The sector weights remain
$1/8,3/8,3/8,1/8$. The sectors $s=1$ and $s=2$ are three dimensional, and their
normalized conditional states have the common matrix
\[
 \rho_s(\lambda)=\frac13
 \begin{pmatrix}
 1&\lambda^2&\lambda^2\\
 \lambda^2&1&\lambda^2\\
 \lambda^2&\lambda^2&1
 \end{pmatrix},
 \qquad s=1,2,
\]
whose spectrum is
\[
 \left(\frac{1+2\lambda^2}{3},
       \frac{1-\lambda^2}{3},
       \frac{1-\lambda^2}{3}\right).
\]
The one-dimensional endpoint sectors remain pure. It follows directly that
\begin{align}
 D_{\mathrm{sector}}(\lambda)&=1-\lambda^2,\label{eq:dephasing-distance}\\
 \Delta S_{\mathrm{sector}}(\lambda)&=\frac34
 H\!\left(\frac{1+2\lambda^2}{3},
             \frac{1-\lambda^2}{3},
             \frac{1-\lambda^2}{3}\right).
 \label{eq:dephasing-entropy}
\end{align}
At $\lambda=1/2$ the two nontrivial sector spectra are
$(1/2,1/4,1/4)$, so
\[
 D_{\mathrm{sector}}=\frac34,
 \qquad
 \Delta S_{\mathrm{sector}}=\frac98\log2\simeq0.77979.
\]
Here $p_*=3/8$, so $D_{\mathrm{sector}}$ lies beyond the first binary
branch. Evaluating the rank-lattice envelope for the full sector profile
\[
 p=(1/8,3/8,3/8,1/8),
 \qquad
 \mathbf n=(1,3,3,1),
\]
for which only the two middle sectors contribute to distance and entropy, gives
\[
 \Phi_{p,\mathbf n}^{-1}\!\left(\frac98\log2\right)
 \simeq0.92115.
\]
The convex surrogate of Section~\ref{sec:convexification} already gives the
nontrivial but weaker certificate
\[
 D_{\mathrm{sector}}
 \le \frac34+\frac{\log2}{8\log(3/2)}
 \simeq0.96369.
\]
By contrast, inversion of the universal three-outcome entropy--error envelope
in Proposition~\ref{prop:feder-merhav-comparison} gives
\[
 D_{\mathrm{sector}}
 \le 1+\frac{\log2}{24\log(3/2)}
 \simeq1.07123,
\]
which is vacuous because the admissible range for this sector profile is
$D_{\max}=1$.  The fixed weights therefore retain information that is lost
under unrestricted classical convexification.  The exact envelope improves
the certificate further to $0.92115$, while the actual distance is $3/4$.
\end{example}

These statements are finite dimensional and apply directly to lattice systems
and finite-dimensional or UV-regularized sector decompositions. Continuum
field-theory constructions of symmetry-resolved entropy, such as
\cite{MurcianoDiGiulioCalabrese2020}, motivate the sector structure but are
not by themselves direct infinite-dimensional applications of the theorem.

\subsection{Projective dephasing}
For a projective dephasing channel
\[
 \mathcal D(\rho)=\sum_iP_i\rho P_i,
\]
the output has the same direct-sum structure. When the block populations are
fixed or otherwise constrained, the weight-profile envelope gives a
finite-dimensional certificate of blockwise purity for the dephased state. Coherent off-diagonal terms $P_i\rho P_j$ with $i\ne j$ require an
additional coherence-sensitive quantity, and dynamical convergence requires a
specified dynamical model.

\section{Comparison with classical and quantum entropy bounds}
\label{sec:comparison}

The sharp Fannes--Audenaert inequality states that if
$T=\frac12\|\rho-\sigma\|_1$, then in dimension $d$,
\[
 |S(\rho)-S(\sigma)|\le h_2(T)+T\log(d-1)
\]
in its sharp range~\cite{Audenaert2007}. This is an upper continuity modulus.
The bounds developed here instead give lower moduli tied to the geometry of
the minimizer set and the constraints.

With the natural logarithms used here, quantum Pinsker's inequality reads
\[
 D(\rho\|\sigma)\ge\frac12\|\rho-\sigma\|_1^2.
\]
Watrous states the inequality using base-two logarithms; converting to natural logarithms gives the form above~\cite{Watrous2018}. It relates relative entropy to the distance between a specified pair of states. Here the controlled quantity is
$S(\rho)-S_{\min}$ and the target is a minimizer set. A pure minimizer need
not even give a finite relative-entropy reference for a nearby mixed state
with larger support.

At the spectral level, Theorem~\ref{thm:one-block-profile} is a classical
entropy--error extremal.  Kovalevsky~\cite{Kovalevsky1968} obtained the sharp
Shannon reverse-Fano lower bound, followed by related formulations and
extensions including Tebbe--Dwyer~\cite{TebbeDwyer1968},
Feder--Merhav~\cite{FederMerhav1994}, Rioul~\cite{Rioul2023}, and the
Arimoto--R\'enyi generalization of Sason--Verd\'u~\cite{SasonVerdu2018}.
Ho--Yeung~\cite{HoYeung2010} studies sharp entropy/variational-distance
relations.  Ho--Verd\'u~\cite{HoVerdu2010} fixes the marginal law of the
variable being estimated and derives tight upper bounds on conditional
entropy at fixed error probability; Prelov~\cite{Prelov2014} considers
entropy extrema when one marginal law and a mismatch probability are fixed.
Hu--Xing~\cite{HuXing2016} derives binary entropy--error bounds directly from
joint-distribution optimization.  These problems retain forms of prior or
joint-distribution information, but they do not prescribe the block mixture
weights and dimensions used here.

Additional posterior information leads to other refinements.  Prasad
\cite{Prasad2015} gives sequences of equivocation bounds based on leading
posterior probabilities and their averages.  The closest structural
comparison is Hled\'ik, Sokolowski and Tka\v{c}ik
\cite{HledikSokolowskiTkacik2019}, who fix the MAP confusion matrix.  Their
minimizing construction splits decode classes into outputs with flat
posteriors; this splitting convexifies the single-posterior reverse-Fano
profile to the piecewise-linear conditional envelope.  In the fixed-block
problem a block carries one density operator with one prescribed weight, so
that splitting operation is not available.  The unconvexified stagewise
profile $\beta_n$ and its rank kinks therefore remain in the optimization.
This produces the value function $\Phi_{q,\mathbf n}$, the rank-lattice edge
geometry, and the possibility of nonnested switching.  The extreme-point
minimum principle used in the finite reduction is standard
\cite{Rockafellar1970}.

\section{Computational use of the bounds}
\label{sec:computational}

For a finite-dimensional block-separable model, the bounds can be evaluated
from the block data:
\begin{enumerate}
\item compute the block minima $m_i$ and minimizer sets $\mathcal M_i$;
\item minimize $F(p)=H(p)+m\cdot p$ over the marginal feasible set; for a
      polytope it suffices to inspect its vertices;
\item at a fixed-support minimizing vertex, compute
      $\kappa_q$ from~\eqref{eq:kappa};
\item for curved constraints, evaluate KKT multipliers and the constrained
      Hessian of the Lagrangian;
\item for fixed populations and unrestricted blocks, construct the
      fixed-population profile by enumerating the rank vertices
      $(D_{\mathbf k},E_{\mathbf k})$ and the adjacent arcs
      $\Gamma_{\mathbf k,j}$ in~\eqref{eq:edge-arc}.
\end{enumerate}
A direct enumeration of the exact profile contains at most
\[
 \sum_j(n_j-1)\prod_{i\ne j}n_i
\]
directed rank edges before pruning.  The exact value can alternatively be
built recursively by min-plus convolution of the one-block profiles,
\[
 \Phi^{(m+1)}(D)
 =\min_y\left\{\Phi^{(m)}(D-y)
 +q_{m+1}\beta_{n_{m+1}}(y/q_{m+1})\right\},
\]
with the corresponding capacity constraints.

For many sectors, Theorem~\ref{thm:convexification} supplies a cheaper
certified surrogate.  Its synchronized vertices
$(\overline D_\ell,\overline E_\ell)$ require no product-rank enumeration:
once the interval containing $D$ is located, a single linear interpolation
gives $\Phi_{q,\mathbf n}^{\mathrm{cvx}}(D)$, together with the explicit
error bound~\eqref{eq:convexification-gap}.  This does not replace the exact
nonconvex envelope when the switching structure itself is required, but it
gives a controlled approximation without enumerating its edges.

\section{Conclusion}
\label{sec:conclusion}

For fixed block populations, the minimum entropy excess at trace-norm
distance $D$ from blockwise purity is $\Phi_{q,\mathbf n}(D)$, the lower
envelope of finitely many rank-lattice edges over the full admissible interval.
Every minimizing distance allocation has all but at most one coordinate at an
integer-rank endpoint.  The binary $q_*$ expression is the first branch;
subsequent branches depend on the complete weight and dimension profiles and
need not form a nested sequence of optimizers.

The lower convex envelope of $\Phi_{q,\mathbf n}$ has a simpler exact
description through synchronized rank vertices.  Its distance from the exact
nonconvex profile is bounded by the sharp uniform constant
$q_\diamond\log(5/4)$, with smaller stage-dependent constants.  For binomial
charge-sector weights this discrepancy is $O(L^{-1/2})$, so the convex
surrogate becomes uniformly accurate in the many-sector limit while retaining
the fixed dimension caps.

For general block-separable constraints, the gap decomposition separates
marginal and within-block stability.  Fixed-support polyhedral minima have
linear leading growth, support opening produces a $t\log(1/t)$ term, and
curved constraints can produce quadratic or higher-order behavior.  In
symmetry-resolved settings, the fixed-weight envelope turns within-sector
entropy into a trace-norm certificate of sectorwise purity.  The three-qubit
phase-damping example shows explicitly that this certificate can remain
nonvacuous after the universal classical entropy--error envelope has lost all
distance information.  Coherent off-diagonal blocks and infinite-dimensional
sectors require additional structure beyond the setting considered here.

\section*{AI use disclosure}
OpenAI ChatGPT (GPT-5.6 Sol, OpenAI) was used during manuscript development
for literature search and synthesis, adversarial proof checking, numerical
stress-test design, and language editing.  These uses were intended to expose
possible errors and identify relevant prior work.  All mathematical
arguments, computations, references, and conclusions were independently
checked by the author, who takes full responsibility for the content.

\section*{Conflict of interest}
The author has no conflicts to disclose.

\section*{Funding}
No external funding was received for this work.

\section*{Data availability}
No data were used in this work.

\appendix

\section{Differential facts for entropy on a fixed support}
\label{app:differential}

Let $\sigma$ be faithful on a finite-dimensional support $P$ and let $X$ be
Hermitian with $\Tr X=0$ and supported on $P$. The entropy is smooth in a
neighbourhood of $\sigma$ within the affine space of trace-one operators on
$P$. Its first differential is
\[
 DS(\sigma)[X]
 =-\Tr\bigl[X(\log\sigma+I)\bigr]
 =-\Tr(X\log\sigma),
\]
since $\Tr X=0$. Its second differential is strictly negative for $X\ne0$;
one representation is
\[
 D^2S(\sigma)[X,X]
 =-\int_0^\infty
 \Tr\!\left[X(\sigma+tI)^{-1}X(\sigma+tI)^{-1}\right]dt<0.
\]
This is the local analytic expression behind
Proposition~\ref{prop:no-zero-slope}.

\bibliographystyle{aipnum4-2}
\bibliography{references}

@book{Watrous2018,
  author    = {Watrous, John},
  title     = {The Theory of Quantum Information},
  publisher = {Cambridge University Press},
  address   = {Cambridge},
  year      = {2018},
  doi       = {10.1017/9781316848142}
}

@book{MarshallOlkinArnold2011,
  author    = {Marshall, Albert W. and Olkin, Ingram and Arnold, Barry C.},
  title     = {Inequalities: Theory of Majorization and Its Applications},
  edition   = {2},
  publisher = {Springer},
  address   = {New York},
  year      = {2011},
  doi       = {10.1007/978-0-387-68276-1}
}

@article{Audenaert2007,
  author  = {Audenaert, Koenraad M. R.},
  title   = {A Sharp Continuity Estimate for the von {N}eumann Entropy},
  journal = {Journal of Physics A: Mathematical and Theoretical},
  volume  = {40},
  number  = {28},
  pages   = {8127--8136},
  year    = {2007},
  doi     = {10.1088/1751-8113/40/28/S18}
}

@book{Rockafellar1970,
  author    = {Rockafellar, R. Tyrrell},
  title     = {Convex Analysis},
  publisher = {Princeton University Press},
  address   = {Princeton, NJ},
  series    = {Princeton Mathematical Series},
  volume    = {28},
  year      = {1970}
}

@book{BonnansShapiro2000,
  author    = {Bonnans, J. Fr{\'e}d{\'e}ric and Shapiro, Alexander},
  title     = {Perturbation Analysis of Optimization Problems},
  publisher = {Springer},
  address   = {New York},
  year      = {2000},
  doi       = {10.1007/978-1-4612-1394-9}
}

@article{BartlettRudolphSpekkens2007,
  author  = {Bartlett, Stephen D. and Rudolph, Terry and Spekkens, Robert W.},
  title   = {Reference Frames, Superselection Rules, and Quantum Information},
  journal = {Reviews of Modern Physics},
  volume  = {79},
  number  = {2},
  pages   = {555--609},
  year    = {2007},
  doi     = {10.1103/RevModPhys.79.555}
}

@article{GoldsteinSela2018,
  author  = {Goldstein, Moshe and Sela, Eran},
  title   = {Symmetry-Resolved Entanglement in Many-Body Systems},
  journal = {Physical Review Letters},
  volume  = {120},
  number  = {20},
  pages   = {200602},
  year    = {2018},
  doi     = {10.1103/PhysRevLett.120.200602}
}

@article{BonsignoriRuggieroCalabrese2019,
  author  = {Bonsignori, Riccarda and Ruggiero, Paola and Calabrese, Pasquale},
  title   = {Symmetry Resolved Entanglement in Free Fermionic Systems},
  journal = {Journal of Physics A: Mathematical and Theoretical},
  volume  = {52},
  number  = {47},
  pages   = {475302},
  year    = {2019},
  doi     = {10.1088/1751-8121/ab4b77}
}

@article{MurcianoDiGiulioCalabrese2020,
  author  = {Murciano, Sara and Di Giulio, Giuseppe and Calabrese, Pasquale},
  title   = {Entanglement and Symmetry Resolution in Two Dimensional Free Quantum Field Theories},
  journal = {Journal of High Energy Physics},
  volume  = {2020},
  number  = {8},
  pages   = {73},
  year    = {2020},
  doi     = {10.1007/JHEP08(2020)073}
}

@article{Jones2022,
  author  = {Jones, Nick G.},
  title   = {Symmetry-Resolved Entanglement Entropy in Critical Free-Fermion Chains},
  journal = {Journal of Statistical Physics},
  volume  = {188},
  pages   = {28},
  year    = {2022},
  doi     = {10.1007/s10955-022-02941-3}
}

@article{MonkmanSirker2023,
  author  = {Monkman, Kyle and Sirker, Jesko},
  title   = {Symmetry-Resolved Entanglement: General Considerations, Calculation from Correlation Functions, and Bounds for Symmetry-Protected Topological Phases},
  journal = {Journal of Physics A: Mathematical and Theoretical},
  volume  = {56},
  number  = {49},
  pages   = {495001},
  year    = {2023},
  doi     = {10.1088/1751-8121/ad086d}
}

@article{FederMerhav1994,
  author  = {Feder, Meir and Merhav, Neri},
  title   = {Relations Between Entropy and Error Probability},
  journal = {IEEE Transactions on Information Theory},
  volume  = {40},
  number  = {1},
  pages   = {259--266},
  year    = {1994},
  doi     = {10.1109/18.272494}
}

@article{HoYeung2010,
  author  = {Ho, Siu-Wai and Yeung, Raymond W.},
  title   = {The Interplay Between Entropy and Variational Distance},
  journal = {IEEE Transactions on Information Theory},
  volume  = {56},
  number  = {12},
  pages   = {5906--5929},
  year    = {2010},
  doi     = {10.1109/TIT.2010.2080452}
}

@article{HoVerdu2010,
  author  = {Ho, Siu-Wai and Verd{\'u}, Sergio},
  title   = {On the Interplay Between Conditional Entropy and Error Probability},
  journal = {IEEE Transactions on Information Theory},
  volume  = {56},
  number  = {12},
  pages   = {5930--5942},
  year    = {2010},
  doi     = {10.1109/TIT.2010.2080891}
}

@misc{ShenChenWang2026,
  author        = {Shen, Haojie and Chen, Jie and Wang, Xiaoqun},
  title         = {Symmetry-Resolved Properties of the Trace Distance in Thermalizing {SU(2)} Systems},
  year          = {2026},
  eprint        = {2603.26540},
  archivePrefix = {arXiv},
  primaryClass  = {quant-ph}
}

@article{TebbeDwyer1968,
  author  = {Tebbe, David L. and Dwyer, III, Samuel J.},
  title   = {Uncertainty and the Probability of Error},
  journal = {IEEE Transactions on Information Theory},
  volume  = {14},
  number  = {3},
  pages   = {516--518},
  year    = {1968},
  doi     = {10.1109/TIT.1968.1054135}
}

@article{Rioul2023,
  author  = {Rioul, Olivier},
  title   = {The Interplay between Error, Total Variation, {Alpha}-Entropy and Guessing: {Fano} and {Pinsker} Direct and Reverse Inequalities},
  journal = {Entropy},
  volume  = {25},
  number  = {7},
  pages   = {978},
  year    = {2023},
  doi     = {10.3390/e25070978}
}

@incollection{Kovalevsky1968,
  author    = {Kovalevsky, V. A.},
  title     = {The Problem of Character Recognition from the Point of View of Mathematical Statistics},
  booktitle = {Character Readers and Pattern Recognition},
  publisher = {Spartan Books},
  address   = {New York},
  pages     = {3--30},
  year      = {1968}
}

@article{Prasad2015,
  author  = {Prasad, Sudhaker},
  title   = {Bayesian Error-Based Sequences of Statistical Information Bounds},
  journal = {IEEE Transactions on Information Theory},
  volume  = {61},
  number  = {9},
  pages   = {5052--5062},
  year    = {2015},
  doi     = {10.1109/TIT.2015.2457913}
}

@inproceedings{HledikSokolowskiTkacik2019,
  author    = {Hled{\'i}k, Michal and Sokolowski, Thomas R. and Tka{\v{c}}ik, Ga{\v{s}}per},
  title     = {A Tight Upper Bound on Mutual Information},
  booktitle = {2019 IEEE Information Theory Workshop (ITW)},
  pages     = {1--5},
  year      = {2019},
  doi       = {10.1109/ITW44776.2019.8989292}
}

@article{Prelov2014,
  author  = {Prelov, V. V.},
  title   = {On One Extreme Value Problem for Entropy and Error Probability},
  journal = {Problems of Information Transmission},
  volume  = {50},
  number  = {3},
  pages   = {203--216},
  year    = {2014},
  doi     = {10.1134/S003294601403016}
}

@unpublished{NasreddineRankStratified2026,
  author = {Nasreddine, Hassan},
  title  = {Rank-Stratified Entropy Rigidity at Quantum State Boundaries},
  note   = {Manuscript},
  year   = {2026}
}

@article{SasonVerdu2018,
  author  = {Sason, Igal and Verd{\'u}, Sergio},
  title   = {{Arimoto--R{\'e}nyi} Conditional Entropy and Bayesian {M}-ary Hypothesis Testing},
  journal = {IEEE Transactions on Information Theory},
  volume  = {64},
  number  = {1},
  pages   = {4--25},
  year    = {2018},
  doi     = {10.1109/TIT.2017.2757496}
}

@article{HuXing2016,
  author  = {Hu, Bao-Gang and Xing, Hong-Jie},
  title   = {An Optimization Approach of Deriving Bounds between Entropy and Error from Joint Distribution: Case Study for Binary Classifications},
  journal = {Entropy},
  volume  = {18},
  number  = {2},
  pages   = {59},
  year    = {2016},
  doi     = {10.3390/e18020059}
}

\end{document}